\documentclass[12pt]{article}
\usepackage{amsmath}
\usepackage{amsthm}
\usepackage{amssymb}
\usepackage{graphicx}
\usepackage{hyperref}

\title{Distribution of codewords on the faces of a hypercube and new combinatorial identities}

\author{Jamolidin K. Abdurakhmanov\\
Department of Information Technologies\\
Andijan State University\\
Andijan, Uzbekistan\\
\texttt{abduraxmanov@adu.uz, jamolidinkamol@gmail.com}}

\date{}

\newtheorem{theorem}{Theorem}
\newtheorem{lemma}{Lemma}
\newtheorem{corollary}{Corollary}
\newtheorem{definition}{Definition}

\begin{document}

\maketitle

\begin{abstract}
We present a novel framework for studying combinatorial identities through the geometric lens of subset distributions in $q$-valued cubes. By analyzing how elements of arbitrary subsets are distributed among the faces of the cube $E_q^n$, we discover new combinatorial identities with geometric significance. We prove that for any subset $A \subset E_2^n$, the rank function satisfies refined bounds that lead to exact computations for small cardinalities. Our main theorem establishes identities connecting the number of $k$-dimensional faces containing exactly $e$ elements of a subset to binomial sums over all subsets of specified cardinality. As applications, we derive a geometric interpretation of Vandermonde's identity and obtain a completely new identity for even-weight vectors. This geometric approach reveals hidden symmetries in classical combinatorial structures and provides a unified framework for generating new identities.
\end{abstract}

\section{Introduction}

The study of combinatorial identities has a rich history dating back to the fundamental works of Euler, Vandermonde, and Cauchy. In recent years, there has been renewed interest in discovering new identities through geometric and algebraic methods, as demonstrated in recent works on generalizations of classical identities \cite{Mesh25} and modern approaches using complex analysis \cite{ChakRoy24}. This paper introduces a novel geometric framework for generating combinatorial identities by analyzing how subsets of $q$-valued cubes are distributed among the various dimensional faces of the cube.

Our approach builds upon the concept of subset rank introduced in our previous work \cite{Abd25}, which measures the dimension of the smallest face containing a given subset. By studying the distribution of subset elements across faces of different dimensions, we uncover a rich class of combinatorial identities that encode geometric properties of the underlying space.

The main contributions of this paper are threefold:
\begin{enumerate}
\item We establish refined bounds for the rank function in binary cubes that enable exact computation for small subsets.
\item We prove a general theorem relating face distributions to binomial sums, yielding a parametric family of identities.
\item We demonstrate that classical identities like Vandermonde's have natural geometric interpretations in our framework, and discover completely new identities.
\end{enumerate}

Throughout this paper, we adopt the convention that binomial coefficients are defined as:
\[\binom{n}{k} = \begin{cases}
\frac{n!}{k!(n-k)!}, & \text{if } k \geq 0, \\
0, & \text{if } k < 0,
\end{cases}\]
where $n, k$ are integers and $n \geq 0$, with the standard convention $0! = 1$.

\section{Preliminaries and Notation}

We begin by recalling the basic definitions and results from \cite{Abd25} that form the foundation of our work.

\begin{definition}
The set of all $n$-dimensional vectors whose coordinates belong to $E_q = \{0, 1, \ldots, q-1\}$ is called the \emph{$n$-dimensional $q$-valued cube} and is denoted by $E_q^n$.
\end{definition}

\begin{definition}
A \emph{$k$-dimensional face} of $E_q^n$ is the set of all vectors in $E_q^n$ with some fixed $n-k$ coordinates.
\end{definition}

The vectors of a subset $A \subset E_q^n$ are also called the vertices, points, or elements of $A$. A fundamental concept in our framework is the rank of a subset:

\begin{definition}
The \emph{rank} of a subset $A \subset E_q^n$, denoted $R(A)$, is the dimension of the smallest face containing $A$.
\end{definition}

This is equivalent to the number of variable columns in the matrix $M_A$ whose rows consist of all vectors in $A$.

Two subsets $A$ and $B$ of $E_q^n$ are called \emph{isometric} if there exists a bijective mapping $\varphi: A \to B$ preserving the Hamming distance:
\[d_H(a, a') = d_H(\varphi(a), \varphi(a'))\]
for all $a, a' \in A$. We denote by $J_A$ the isometry class containing $A$.

For a subset $A$, we define $D_A$ as the sum of Hamming distances between all distinct pairs of vectors in $A$:
\[D_A = \sum_{\substack{x, y \in A \\ x \neq y}} d_H(x, y).\]

\section{Refined Bounds for Binary Cubes}

Our first main result refines the bounds for $R(A)$ established in \cite{Abd25} for the special case of binary cubes.

\begin{theorem}\label{thm:refined-bounds}
For any subset $A \subset E_2^n$ with $|A| \geq 1$ and any $B \in J_A$, the following relations hold:
\begin{enumerate}
\item[1)] $R(B) = 0$ if $|A| = 1$,
\item[2)] $\frac{4D_A}{|A|^2} \leq R(B) \leq \frac{D_A}{|A|-1}$ if $|A| > 1$ and $|A|$ is even,
\item[3)] $\frac{4D_A}{|A|^2-1} \leq R(B) \leq \frac{D_A}{|A|-1}$ if $|A| > 1$ and $|A|$ is odd.
\end{enumerate}
\end{theorem}

\begin{proof}
The case $|A| = 1$ is trivial. For $|A| > 1$, since isometric subsets have the same rank and distance sum, it suffices to prove the bounds for $A$ itself.

By Theorem 3.1 in \cite{Abd25}, for $q = 2$ and $|A| \geq 2$, we have:
\[\frac{4D_A}{|A|^2} \leq R(A) \leq \frac{D_A}{|A|-1}.\]

For odd $|A|$, we need to prove the stronger lower bound. Consider the matrix $M_A$ whose rows are the vectors of $A$. Let the first $R(A)$ columns be variable, with column $j$ having $x_j$ zeros and $y_j$ ones, where $x_j + y_j = |A|$. Then:
\[D_A = \sum_{j=1}^{R(A)} x_j y_j.\]

When $|A|$ is odd, to maximize $D_A$, we must have $\{x_j, y_j\} = \{(|A|-1)/2, (|A|+1)/2\}$ for each $j$. This gives:
\[D_A \leq R(A) \cdot \frac{|A|^2-1}{4},\]
which yields the desired lower bound.
\end{proof}

This theorem allows exact computation of ranks for small subsets. For $A \subset E_2^n$:
\begin{itemize}
\item If $A = \{a_1\}$, then $R(A) = 0$.
\item If $A = \{a_1, a_2\}$, then $R(A) = d_{12}$, where $d_{12} = d_H(a_1, a_2)$.
\item If $A = \{a_1, a_2, a_3\}$, then $R(A) = (d_{12} + d_{13} + d_{23})/2$, where $d_{ij} = d_H(a_i, a_j)$.
\end{itemize}

\section{Distribution of Elements in Faces}

To state our main results on face distributions, we first recall the concept of multisets. A \emph{multiset} on a finite set $U$ is the set $U$ together with a multiplicity function $r: U \to \mathbb{N}_0$. The cardinality of a multiset is $\sum_{a \in U} r(a)$. The sum of multisets $U_1, \ldots, U_l$ with multiplicity functions $r_1, \ldots, r_l$ is the multiset with multiplicity function $r(a) = r_1(a) + \cdots + r_l(a)$, denoted $\biguplus_{i=1}^l U_i$.

Let $v(A, k)$ denote the number of $k$-dimensional faces of $E_q^n$ containing the subset $A$.

\begin{lemma}\label{lem:face-count}
For any non-empty subset $A \subset E_q^n$ and integer $k$ with $0 \leq k \leq n$:
\[v(A, k) = \binom{n - R(A)}{k - R(A)}.\]
\end{lemma}

\begin{proof}
If $k < R(A)$, then $v(A, k) = 0$ and the formula holds. For $k \geq R(A)$, to specify a $k$-dimensional face containing $A$, we must include all $R(A)$ variable positions and choose $k - R(A)$ additional positions from the remaining $n - R(A)$ constant positions. The result follows.
\end{proof}

Let $V(A, k, e)$ denote the number of $k$-dimensional faces containing exactly $e$ elements of $A$.

\begin{theorem}\label{thm:main-identity}
Given a non-empty subset $A \subset E_q^n$, for any integers $k$ and $s$ with $0 \leq k \leq n$, $1 \leq s \leq p(k)$, where $p(k) = \min\{|A|, q^k\}$:
\[\sum_{e=s}^{p(k)} \binom{e}{s} V(A, k, e) = \sum_{B \subset A, |B|=s} \binom{n - R(B)}{k - R(B)},\]
where the right sum is over all $s$-element subsets $B \subset A$.
\end{theorem}

\begin{proof}
Let $\Gamma(B, k)$ be the set of all $k$-dimensional faces containing $B$. Consider the multiset:
\[\Gamma' = \biguplus_{B \subset A, |B|=s} \Gamma(B, k).\]

In $\Gamma'$, each $k$-dimensional face containing exactly $e$ elements of $A$ has multiplicity $\binom{e}{s}$, so:
\[|\Gamma'| = \sum_{e=s}^{p(k)} \binom{e}{s} V(A, k, e).\]

On the other hand:
\[|\Gamma'| = \sum_{B \subset A, |B|=s} |\Gamma(B, k)| = \sum_{B \subset A, |B|=s} \binom{n - R(B)}{k - R(B)},\]
using Lemma \ref{lem:face-count}.
\end{proof}

\section{Combinatorial Identities from Geometric Structures}

Theorem \ref{thm:main-identity} provides a general framework for deriving combinatorial identities. We now present several corollaries and examples.

\begin{corollary}\label{cor:small-s}
Let $A \subset E_2^n$, $k$ be an integer with $0 \leq k \leq n$, and $p(k) = \min\{|A|, 2^k\}$. Then:
\begin{enumerate}
\item[1)] For $|A| \geq 1$:
\[\sum_{e=1}^{p(k)} e \cdot V(A, k, e) = |A| \binom{n}{k}.\]

\item[2)] For $|A| \geq 2$:
\[\sum_{e=2}^{p(k)} \binom{e}{2} V(A, k, e) = \sum_{i=1}^{|A|-1} \sum_{j=i+1}^{|A|} \binom{n - d_{ij}}{k - d_{ij}}.\]

\item[3)] For $|A| \geq 3$:
\[\sum_{e=3}^{p(k)} \binom{e}{3} V(A, k, e) = \sum_{i=1}^{|A|-2} \sum_{j=i+1}^{|A|-1} \sum_{t=j+1}^{|A|} \binom{n - d_{ijt}}{k - d_{ijt}},\]
where $d_{ijt} = (d_{ij} + d_{it} + d_{jt})/2$, where $d_{ij}$, $d_{it}$, $d_{jt}$ -- Hamming distances between corresponding elements $a_i$, $a_j$ and $a_t$ of $A$.
\end{enumerate}
\end{corollary}

These identities remain valid for general $q > 2$ in cases 1) and 2).

\subsection{Example: Vandermonde's Identity from Face Geometry}

Let $A \subset E_q^n$ be a $\nu$-dimensional face. The values of $V(A, k, e)$ can be computed explicitly:
\begin{itemize}
\item For $e = 0$: $V(A, k, e) = \sum_{i=0}^{\nu} q^{\nu-i} \binom{\nu}{i} \binom{n-\nu}{k-i} (q^{n-\nu-k+i} - 1)$.
\item For $e = q^i$, $i = 1, 2, \ldots, \nu$: $V(A, k, e) = q^{\nu-i} \binom{\nu}{i} \binom{n-\nu}{k-i}$.
\item For other values of $e$: $V(A, k, e) = 0$.
\end{itemize}

Substituting into Corollary \ref{cor:small-s} part 1) with $|A| = q^\nu$:
\[\sum_{i=0}^{\nu} \binom{\nu}{i} \binom{n-\nu}{k-i} = \binom{n}{k},\]
which is Vandermonde's identity.

For part 2), the number of two-element subsets $B \subset A$ with $R(B) = i$ is $\frac{1}{2}(q-1)^i \binom{\nu}{i} q^\nu$. This leads to the new identity:
\[\sum_{i=1}^{\nu} (q^i - 1) \binom{\nu}{i} \binom{n-\nu}{k-i} = \sum_{i=1}^{\nu} (q-1)^i \binom{\nu}{i} \binom{n-i}{k-i}.\]

This generalizes results on Chu-Vandermonde variations studied by Meštrović \cite{Mesh25}.

\subsection{Example: Even-Weight Vectors}

Let $A \subset E_2^n$ consist of all even-weight vectors (those with an even number of 1's). Then:
\[V(A, k, e) = \begin{cases}
2^{n-k} \binom{n}{k}, & \text{if } e = 2^{k-1}, \\
0, & \text{if } e \neq 2^{k-1}.
\end{cases}\]

From Corollary \ref{cor:small-s} part 2):
\[(2^{k-1} - 1) 2^{n-1} \binom{n}{k} = \sum_{i=1}^{\lfloor n/2 \rfloor} \binom{n}{2i} \binom{n-2i}{k-2i}.\]

This is a completely new combinatorial identity valid for all $1 \leq k \leq n$.

\section{Applications and Connections}

The geometric framework developed here has several important applications:

\subsection{Linear Codes}
When $q$ is a prime power, $E_q^n$ becomes a vector space over $\mathrm{GF}(q)$. Linear codes are subspaces, and our results provide new tools for analyzing their weight distributions. The distinction between algebraic dimension and combinatorial rank offers insights into code structure. Recent work on graph covers \cite{BFHP25} shows similar geometric approaches can yield unexpected results in combinatorial structures.

\subsection{Matroid Theory}
The rank function $R(A)$ defines a matroid on subsets of $E_q^n$. Our identities translate to statements about the characteristic polynomial and Tutte polynomial of these matroids. This connects to recent developments in matroid polynomials, such as the work by Ferroni and Larson \cite{FL24} on Kazhdan-Lusztig polynomials of braid matroids.

\subsection{Quantum Error Correction}
Binary and $q$-ary quantum codes correspond to special subsets of $E_q^n$. Our framework provides new approaches to constructing codes with prescribed distance properties.

\section{Conclusion}

We have introduced a geometric framework for studying combinatorial identities through the distribution of subset elements in faces of $q$-valued cubes. This approach reveals that many classical identities have natural geometric interpretations and enables the discovery of new identities. The connection to enumerative combinatorics, as systematically developed in Stanley's comprehensive treatment \cite{Stanley23}, suggests many further directions for exploration.

The main contributions include refined bounds for subset ranks in binary cubes, a general theorem connecting face distributions to binomial sums, and the discovery of new identities for even-weight vectors and generalizations of Vandermonde's identity.

Future directions include:
\begin{itemize}
\item Extending the framework to infinite-dimensional spaces
\item Developing algorithms for computing $V(A, k, e)$ for general subsets
\item Finding $q$-analogue versions of our identities
\item Applications to extremal combinatorics and coding theory
\end{itemize}

The geometric perspective on combinatorial identities opens new avenues for research at the intersection of combinatorics, geometry, and algebra.


\begin{thebibliography}{99}

\bibitem{Abd25} J. K. Abdurakhmanov, The combinatorial rank of subsets: Metric density in finite Hamming spaces, \emph{arXiv preprint arXiv:2506.13081}, 2025.

\bibitem{Mesh25} R. Meštrović, Several generalizations and variations of Chu-Vandermonde identity, \emph{J. Integer Seq.} \textbf{28} (2025), Article 25.1.3.

\bibitem{ChakRoy24} D. Chakraborty and A. Roy, Vandermonde's identity proved by complex analysis, \emph{Elem. Math.} \textbf{79} (2024), 145--148.

\bibitem{BFHP25} M. Bonamy, P. Fleischmann, T. Huynh, and S. Pelekis, Small even covers of the complete graph, \emph{Electron. J. Combin.} \textbf{32}(1) (2025), \#P1.15.

\bibitem{FL24} L. Ferroni and M. T. L. Larson, Kazhdan-Lusztig polynomials of braid matroids, \emph{J. Lond. Math. Soc.} \textbf{109}(3) (2024), e12875.

\bibitem{Stanley23} R. P. Stanley, \emph{Enumerative Combinatorics}, Volume 1, 2nd ed., Cambridge University Press, 2023.

\end{thebibliography}
\end{document}